\documentclass[copyright,creativecommons]{eptcs}
\providecommand{\event}{CL\&C'14}
\usepackage{times}
\usepackage{color}

\usepackage{amssymb,amsmath,amsthm}
\usepackage[curve, matrix, arrow]{xy}
\usepackage{latexsym}
\usepackage{proof}

\newtheorem{thm}{Theorem}

\newtheorem{cor}[thm]{Corollary}

\newcommand{\mt}{\tilde{\mu}}
\newcommand{\com}[2]{\langle #1|#2\rangle}
\newcommand{\ov}[1]{\overline{#1}}

\newcommand{\lb}{\lambda}

\newcommand{\ehole}{[\,]}

\newcommand{\Eta}{\mathsf{DNeg}}

\newcommand{\up}[1]{{\uparrow}{#1}}

\newcommand{\bindraw}{\mathsf{bind}}
\newcommand{\bind}[3]{\bindraw(#1,#2.#3)}
\newcommand{\letraw}{\mathsf{let}}
\newcommand{\letvc}[3]{\letraw(#1,#2.#3)} 
\newcommand{\ret}{\mathsf{ret}\,}

\def\pcomm(#1|#2){\langle{#1}|{#2}\rangle} 

\def\mutil{{\tilde{\mu}}}

\def\slet{{\sf let}}
\def\sin{{\sf in}}
\def\llet<#1=#2,#3>{{\slet\ {#1}={#2}\ \sin\ {#3}}}

\newcommand{\esubst}[3]{\{#1/#2\}#3}

\newcommand{\subst}[3]{\esubst {#1}{#2}{#3}}

\newcommand{\ol}{\overline{\lambda} }

\newcommand{\lmmt}{\ol\mu\tilde{\mu}}
\newcommand{\lm}{\lambda\mu}
\newcommand{\mlm}{\lm_\mathsf{M}}

\newcommand{\lvc}{{\mathsf{VC}}\mu_{\mathsf{M}}}

\def\vlmm{{\ov{\lambda}\mu\tilde{\mu}^{\sf v}}}

\def\l2{{\lambda^{\sf 2L}}}

\def\nlmm{{\ov{\lambda}\mu\tilde{\mu}^{\sf n}}}

\newcommand{\lmmtvn}{\lmmt_{\mathsf{vn}}}

\newcommand{\mon}{M}

\newcommand{\imp}{\supset}

\newcommand{\seql}[4]{#1\mid #2:#3\vdash #4}

\newcommand{\cps}[1]{\langle\![{#1}]\!\rangle}

\newcommand{\cmmin}[1]{{#1}^{\bullet-}}

\def\vmo#1{{\ov{{#1}}}}

\def\cm#1{{{#1}^\bullet}}

\def\cmaux#1{{{#1}^\star}}

\newcommand{\mtr}[1]{\ov{#1}} 
\newcommand{\mdtr}[1]{{#1}^{\dagger}} 
\newcommand{\mmdtr}[1]{\ov{#1}^{\dagger}}

\newcommand{\sigmabi}{{\sigma}}   

\newcommand{\pib}{\pi_\bindraw}
\newcommand{\pil}{\pi_\letraw}
\newcommand{\pic}{\pi_{\mathsf{covar}}}
\newcommand{\etab}{{\eta_\bindraw}}
\newcommand{\etal}{{\eta_\letraw}}
\newcommand{\etam}{{\eta_\mu}}
\newcommand{\etamt}{{\eta_{\mt}}}
\newcommand{\etamtx}{\eta_{\mutil,x}}

\def\red{\to}

\def\vbeta{{\beta^{\sf v}}}

\newcommand{\sigmav}{\sigma^{\sf v}}
\newcommand{\pin}{\pi^{\sf n}}

\newcommand{\betae}{{\beta_e}}

\begin{document}

\title{Confluence for classical logic through the distinction between values
and computations}

\author
{Jos\'e Esp\'\i{}rito Santo
\institute{Centro de Matem\'{a}tica, Universidade do Minho, Portugal
\email{jes@math.uminho.pt}}
\and
Ralph Matthes
\institute{I.R.I.T. (C.N.R.S. and University of Toulouse III), France
\email{matthes@irit.fr}}
\and
Koji Nakazawa
\institute{Graduate School of Informatics, Kyoto University, Japan
\email{knak@kuis.kyoto-u.ac.jp}}
\and
Lu\'\i{}s Pinto
\institute{Centro de Matem\'{a}tica, Universidade do Minho, Portugal
\email{luis@math.uminho.pt}}
}

\def\titlerunning{Confluence for classical logic through the distinction
between values and computation}
\def\authorrunning{J.~Esp\'\i{}rito~Santo, R.~Matthes, K.~Nakazawa, and L.~Pinto}

\maketitle

\begin{abstract}
We apply an idea originated in the theory of programming languages---monadic
meta-language with a distinction between values and
computations---in the design of a calculus of cut-elimination for
classical logic.  The cut-elimination calculus we obtain comprehends
the call-by-name and call-by-value fragments of Curien-Herbelin's
$\lmmt$-calculus without losing confluence, and is based on a
distinction of ``modes'' in the proof expressions and ``mode''
annotations in types. Modes resemble colors and polarities, but are
quite different: we give meaning to them in terms of a monadic
meta-language where the distinction between values and computations
is fully explored. This meta-language is a refinement of the
classical monadic language previously introduced by the authors, and
is also developed in the paper.
\end{abstract}


\section{Introduction}\label{sec:intro}

It is well-known that confluence fails for cut elimination in
classical logic in the worst way: proof identity is trivialized
\cite{GirardLafontTaylor89}. Computationally, this trivialization is
caused by the ``superimposition'' of call-by-name (cbn) and
call-by-value (cbv) in the proof reduction of classical sequent
calculus \cite{CurienHerbelinICFP2000,HerbelinHabilitationThesis}.

Several solutions have been proposed to this problem
(e.g., \cite{GirardMSCS91,DanosJoinetSchellinxJSL97,CurienHerbelinICFP2000,ZeilbergerAPAL08,Munch-MaccagnoniCSL09,Curien-Munch-Maccagnoni-ifipTCS10}).
Some solutions consist in constraining the set of derivations
(e.g., the sequent calculi $LC$ \cite{GirardMSCS91} or $LKT$, $LKQ$
\cite{CurienHerbelinICFP2000}), others constrain the reduction rules
(e.g., the cbn and cbv fragments of $\lmmt$
\cite{CurienHerbelinICFP2000}). A third kind relies on the
enrichment/refinement of the syntax of formulas, by means of colors
or polarities
\cite{DanosJoinetSchellinxJSL97,ZeilbergerAPAL08,Munch-MaccagnoniCSL09,Curien-Munch-Maccagnoni-ifipTCS10}.
We propose a new solution of the latter kind: a confluent variant of
the system $\lmmt$ \cite{CurienHerbelinICFP2000} that comprehends
the cbn and cbv fragments of $\lmmt$, and that is based on a
distinction of two ``modes'' in the proof expressions and ``mode''
annotations in function spaces. Our system is called \emph{$\lmmt$
with modes} and denoted $\lmmtvn$.

In a self-contained explanation of $\lmmtvn$, one starts by
splitting the set of variables in proof expressions into two
disjoint sets, by singling out a set of ``value variables''. This
allows a refinement of the notions of value and co-value which
immediately solves the cbn/cbv dilemma. However, $\lmmtvn$ has many
design decisions that may look peculiar at first sight. For
instance, atomic formulas do not get a mode annotation, while
composite formulas do; and while variables in proof expressions
get a mode, co-variables do not.

A full semantics for (the design of) $\lmmtvn$  is given in terms of
a monadic meta-language of the kind introduced by Moggi
\cite{Moggi91}. This meta-language, called the \emph{calculus of
values and computations} and denoted $\lvc$, is also developed in
the present paper. It is a refinement of the monadic language
previously introduced by the authors \cite{ourMSCS}, and as such
combines classical logic with a monad. The refinement is guided by
the idea of extending in a coherent way to proof expressions the
distinction between value types and computation types (so that, for
instance, a typable expression is a value iff it is typable with a
value type). In such a system we can interpret the distinction
between the two modes of $\lmmtvn$ in terms of the distinction
value/computation.

Confluence for typed expressions of $\lmmtvn$ is obtained (through
Newman's lemma) from strong normalization. The latter, in turn, is
obtained by proving that two translations produce strict simulation
by strongly normalizing targets: one is the map from $\lvc$ to the
simply-typed $\lb$-calculus induced by instantiating the monad of
$\lvc$ to the continuations monad; the other is the monadic
semantics from $\lmmtvn$ to $\lvc$. As a side remark, we observe
that the composition of the two translations produces a CPS
translation of $\lmmtvn$ which is therefore uniform for cbn and cbv,
since the cbn and cbv fragments of $\lmmt$ are included in
$\lmmtvn$.

\paragraph{Structure of the paper.}
Section~\ref{sec:monadic-langs} recalls the monadic meta-language $\mlm$
\cite{ourMSCS} and develops the calculus of values and
computations $\lvc$. Section \ref{sec:classical-logic} recalls
$\lmmt$ and its main critical pair (the cbn/cbv dilemma), and develops
the proposed variant of $\lmmt$ with modes. Section~\ref{sec:concl}
concludes, and discusses related and future work. See Fig.~\ref{fig:new-picture} for an overview.

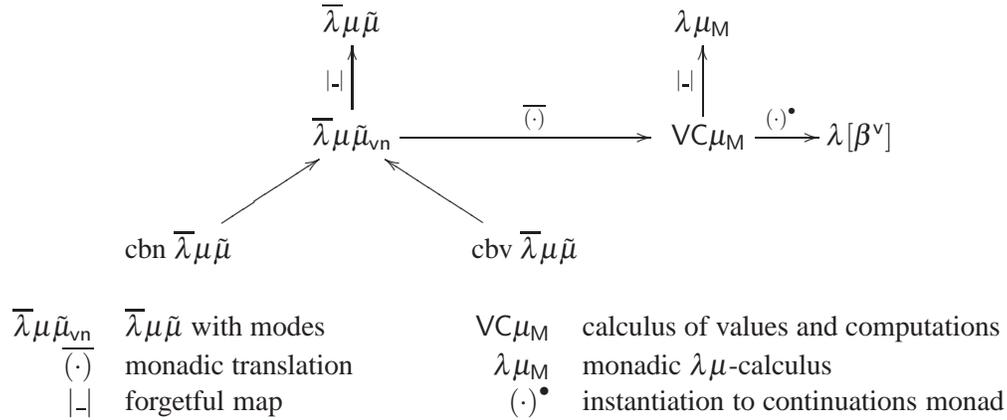
\begin{figure}[t]\caption
{Overview}\label{fig:new-picture}
\begin{center}
\begin{tabular}{c}
\xymatrix{%
&\lmmt&&\mlm&\\
&\lmmtvn\ar[u]^{|\_|}\ar[rr]^{\;\;\;\vmo{(\cdot)}}&&\;\;\lvc\ar[u]^{|\_|}\ar[r]^{{(\cdot)}^{\bullet}}&{\lb[\vbeta]}\\
\textrm{cbn }\lmmt\ar[ru]&&\textrm{cbv }\lmmt\ar[lu]&&\\
}%
\\ \\
\begin{tabular}[t]{rlcrl}
$\lmmtvn$&$\lmmt$ with modes&\qquad\qquad&$\lvc$&calculus of
values and computations\\
$\vmo{(\cdot)}$&monadic translation&\qquad\qquad&$\mlm$&monadic $\lm$-calculus\\
$|\_|$&forgetful map&&$\cm{(\cdot)}$&instantiation to continuations monad\\
\end{tabular}
\end{tabular}
\end{center}
\end{figure}


\section{Monadic meta-languages}\label{sec:monadic-langs}

We start this section by recalling the $\mlm$-calculus. Next we
motivate and formally develop, as a sub-calculus of $\mlm$, a
calculus of values and computations, denoted $\lvc$. We continue
with a comparison between the two monadic languages, and spell out
the intuitionistic fragment of $\lvc$. Finally, we study the map
from $\lvc$ into the simply-typed $\lb$-calculus obtained by
instantiating the monad of $\lvc$ to the continuations monad.

\subsection{The $\mlm$-calculus}

We recapitulate the $\mlm$-calculus that has been proposed by the
present authors \cite{ourMSCS}.

Expressions $T$ are values,
terms, and commands that are defined by the following
grammar\footnote{In the notation of our previous paper
\cite{ourMSCS} $\ret t$ is written $\eta t$, and in the notation of
Moggi \cite{Moggi91}, $\bind txc$ and $\ret t$ are written
\mbox{${\sf let}\, x=t\, {\sf
in}\, c$} and $[t]$, respectively.}:\\[-2ex]
\[
 V  ::= x\mid\lambda x.t
\qquad r,s,t,u ::= V\mid tu\mid\mu a.c\mid\ret t
\qquad c ::= at\mid\bind{t}{x}{c}.
\]
Variable occurrences of $x$ in $t$ of $\lambda x.t$ and $c$ of
$\bind{t}{x}{c}$, and $a$ in $c$ of $\mu a.c$ are bound.
We introduce base contexts $L$ as commands with a ``hole'' for a term of
the following two forms: $a\ehole$ and $\bind{\ehole}{x}{c}$.  For a
term $t$, $L[t]$ is defined as ``hole-filling''.
Term substitution $[s/x]T$ is defined in the obvious way as for
$\lmmt$.
Furthermore, structural substitutions $[L/a]T$ are defined by
recursively replacing every subexpression $au$ of $T$, by $L[u]$ (this
may need renaming of bound variables).
It corresponds to the substitution of co-variables in $\lmmt$.

Types are given by $A,B\,::=\,X\mid A\supset B\mid\mon A$. Thus,
besides the type variables and implication of $\lmmt$, we have a
unary operation $M$ on types. Types of the form $\mon A$ are called
\emph{monadic types}.
Sequents are written as:
$\Gamma\vdash t:A\mid \Delta$ and $c:(\Gamma\vdash\Delta)$. In both
cases, $\Delta$ is a consistent set of declarations $a:\mon A$,
hence with monadic types.
Typing rules and reduction rules are given in
Fig.~\ref{fig:monadic-lambda-munew}.
\begin{figure}[tb]
\caption{Typing rules and reduction rules of $\mlm$}
\label{fig:monadic-lambda-munew}
$$
\infer[Ax]{\Gamma,x:A\vdash x:A\mid \Delta}{}\qquad
\infer[\mathit{Intro}]{\Gamma\vdash\lambda x.t:A\supset
B\mid \Delta}{\Gamma,x:A\vdash t:B\mid \Delta}\qquad \infer[\mathit{Elim}]{\Gamma\vdash
tu:B\mid \Delta}{\Gamma\vdash t:A\supset
 B\mid \Delta&\Gamma\vdash u:A\mid \Delta}
$$
$$
\infer[Pass]{at:(\Gamma\vdash a:\mon A,\Delta)}{\Gamma\vdash t:\mon
A\mid a:\mon A,\Delta}\qquad \infer[Act]{\Gamma\vdash \mu a.c:\mon
A\mid \Delta}{c:(\Gamma\vdash a:\mon A,\Delta)}
$$
$$
\infer[]{\Gamma\vdash \ret s:\mon A\mid \Delta}{\Gamma\vdash s:A\mid \Delta}
\qquad
 \infer[]{\bind rxc:(\Gamma\vdash\Delta)}{\Gamma\vdash r:\mon
A\mid \Delta&c:(\Gamma,x:A\vdash\Delta)}
$$
$\mathstrut$
$$
\begin{array}{r@{\quad}rcl@{\qquad\qquad}r@{\quad}rcl}
(\beta)  & (\lambda x.t)s& \rightarrow& [s/x]t
& (\etam)  & \mu a.at &\rightarrow & t \quad(a\notin t)\\
(\sigmabi) & \bind{\ret s}xc&\rightarrow&[s/x]c
& (\etab) & \bind{t}{x}{a(\ret x)}&\rightarrow&at\\
(\pi)  & L[\mu a.c]&\rightarrow&[L/a]c\\
\end{array}
$$
\end{figure}
Notice that the rule $\pi$ uses the derived syntactic class of base
contexts and is therefore a scheme that stands for the following two
rules\\[-2ex]
\[
\begin{array}{rcrcl}
(\pib) & \quad & \bind{\mu a.c}{x}{c'}&\rightarrow & [\bind{\ehole}{x}{c'}/a]c\\
(\pic) & \quad & b(\mu a.c) &\rightarrow & [b/a]c.
\end{array}
\]
It is easy to see that $\mlm$ satisfies subject reduction, for
strong normalization see our previous paper \cite{ourMSCS}.

\subsection{Towards a calculus of values and computations}

We identify a sub-language of $\mlm$, called $\lvc$, the calculus of
values and computations. The calculus can be motivated as a sharp
implementation of the principles at the basis of Moggi's semantics
of programming languages \cite{Moggi91}.

According to Moggi, each programming language type $\mathtt{A}$
gives rise to the types $A$ (of ``values of type $\mathtt{A}$'') and
$MA$ (of ``computations of type $\mathtt{A}$''). In addition, a
program of type $\mathtt{A}\to\mathtt{A'}$ corresponds to an
expression of type $A\imp MA'$. In this rationale: (i) attention is
paid only to a part of the function space, and (ii) there is no role
to types $M(MA)$, written $M^2A$ in the sequel. $\lvc$ implements
(i), extracting the full consequences at the level of expressions;
(ii) is already realized in any monadic language, since each
expression of type $M^2A$ may be coerced to one of type $MA$ by
monad multiplication, but $\lvc$ goes farther by removing $M^2A$
from the syntax of types.

$\lvc$ is thus obtained from $\mlm$ after three steps of
simplification as follows.

Firstly, we restrict implications to the form $A\supset MA'$. This is
already done, for instance, in the presentation of the monadic
meta-language by Hatcliff and Danvy \cite{HatcliffDanvy94}; however, we do the
restriction in a formal way, by separating a class of types
$C::=MA$. If $B$ denotes a non-monadic type, then types are given by
$A::=\,B\mid C$, with $B::=\,X\mid A\imp C$. Following \emph{op.\ cit.}, we call types $B$ (resp.~$C$) ``value types'' (resp.~``computation types'').

Secondly, we pay attention to expressions. We now have two meanings
for the word ``value'': either as a term with value type, or the
``traditional'' one of being a variable or $\lambda$-abstraction. So
far, a term has value type only if it is a ``traditional value''. On
the other hand, the separation into value and computation types
splits the term-formers into $\lb$-abstraction (with value type) and
$\ret t$, $tu$ and $\mu a.c$ (with computation type). The full split
of terms into two categories, \emph{values} $V$ and
\emph{computations} $P$, is obtained by separating two sets of term
variables, \emph{value variables} $v$ and \emph{computation
variables} $p$, with the intention of having a \emph{well-moded}
typing system, that is, one that assigns to the variables types with
the right ``mode'' (value or computation). This achieves coherence
for the two meanings of ``value'': a term has a value type iff it is
a traditional value (that is, a \emph{value} variable or
$\lambda$-abstraction). It follows that a term has a computation
type iff it is a computation. At this point, we are sure not to lose
any typable terms, if we restrict $tu$ and $\lambda x.t$ to $Vu$ and
$\lb x.P$, respectively.

Thirdly, we restrict computation types (hence the type of
co-variables) to $MB$, thus forbidding $M^2A$, and forcing $\ret V$
instead of $\ret t$.

The formal presentation of $\lvc$ follows.

\subsection{The calculus $\lvc$}

\textbf{Expressions.}
The variables of $\mlm$ are divided into two disjoint name spaces, and
denoted by $x$ if any of them is meant. Co-variables are ranged over
by $a$, $b$, as for $\lmmt$ and $\mlm$. Expressions are given by the
grammar in Fig.~\ref{fig:lvc}.
\begin{figure}[tb]\caption{Expressions and types of $\lvc$}\label{fig:lvc}
\[
\begin{array}{rrclrrcl}
(\textrm{value vars}) & v,w,f&&&(\textrm{value types}) & B & ::= & X\mid A \imp C\\
(\textrm{comp. vars}) & n,p,q&&&(\textrm{comp.~types}) & C & ::= & MB\\
(\textrm{variables})& x,y,z&::=& v\mid p&(\textrm{types}) & A  & ::= & B\mid C\\
(\textrm{values}) & V,W & ::= & v\mid\lambda x.P\\
(\textrm{computations}) & P,Q & ::= & p\mid\ret V\mid Vu\mid\mu a.c\\
(\textrm{terms}) & t,u & ::= & V\mid P\\
(\textrm{commands}) & c & ::= & aP\mid\letvc{P}{v}{c}\mid\subst Ppc
\end{array}
\]
\end{figure}

\textbf{Types.} The motivation for these syntactic distinctions
comes from the types that should be assigned. The type system of
$\mlm$ is also restricted and divided into two classes, see again
Fig.~\ref{fig:lvc}.
In particular, as explained before, there is no type of the form $M(MA)$ in  $\lvc$.

The idea of the distinction into values and computations is that
values receive value types and computations receive computation
types in a context where value variables are assigned value types
and computation variables are assigned computation types. Such
contexts will be called \emph{well-moded}. It is remarkable that the
distinction can be done on the level of raw syntax (and that it will
be preserved under the reduction rules to be presented below). The
new syntax element $\subst Ppc$ represents $\bind{\ret P}pc$ in
$\mlm$. This means that no argument $t$ to $\bind tpc$ other than of
the form $\ret P$ is considered, but this is not seen as composed of
a $\bindraw$ and a $\ret$ operation but atomic in $\lvc$. The expression
$\ret P$ does not even belong to $\lvc$.  See Section
\ref{subseq:bind-let-subst} for more on the connection with $\mlm$.

\textbf{Typing rules} are inherited from $\mlm$, with their full
presentation in Fig.~\ref{fig:typing-rules-CVC}.
Here, every context $\Gamma$ in the judgements is
\emph{well-moded} in the sense given above. As for $\mlm$, the
contexts $\Delta$
consist only of bindings of the
form $a:\mon A$, which, for $\lvc$ even requires $a:\mon B$, hence
$a:C$. Thus, more precisely, co-variables might be called
``computation co-variables''.

\begin{figure}[tb]\caption{Typing rules of $\lvc$}\label{fig:typing-rules-CVC}
$$
\begin{array}{c}
\infer[Axv]{\Gamma,v:B\vdash v:B\mid \Delta}{}\qquad
\infer[Axc]{\Gamma,p:C\vdash p:C\mid \Delta}{}\\[1ex]
\infer[\mathit{Intro}]{\Gamma\vdash\lambda x.P:A\supset
C\mid \Delta}{\Gamma,x:A\vdash P:C\mid \Delta}\qquad
\infer[\mathit{Elim}]{\Gamma\vdash Vu:C\mid \Delta}{\Gamma\vdash V:A\supset C\mid \Delta&\Gamma\vdash u:A\mid \Delta}\\[1ex]
\infer[Pass]{aP:(\Gamma\vdash a:C,\Delta)}{\Gamma\vdash
P:C\mid a:C,\Delta}\qquad\infer[Act]{\Gamma\vdash \mu
a.c:C\mid \Delta}{c:(\Gamma\vdash
a:C,\Delta)}\\[1ex]
\infer{\Gamma\vdash \ret V:\mon B\mid \Delta}{\Gamma\vdash V:B\mid \Delta}
\qquad \infer{\letvc Pvc:(\Gamma\vdash\Delta)}{\Gamma\vdash
P:\mon B\mid \Delta&c:(\Gamma,v:B\vdash\Delta)}\\[1ex]
\infer{\subst Ppc:(\Gamma\vdash\Delta)}{\Gamma\vdash
P:C\mid \Delta&c:(\Gamma,p:C\vdash\Delta)}
\end{array}
$$
\end{figure}

Clearly, the above-mentioned intuition can be made precise in that
$\Gamma\vdash P:A\mid \Delta$ implies that $A$ is a computation type and
that $\Gamma\vdash V:A\mid \Delta$ implies that $A$ is a value type.
This can be read off immediately from
Fig.~\ref{fig:typing-rules-CVC}.

\emph{Well-moded} substitutions $[u/x]t$ and $[u/x]c$, i.\,e., with
$x$ and $u$ either value variable and value or computation variable
and computation, are inherited from $\mlm$. Well-moded substitution
$[u/x]t$ respects modes in that $[u/x]V$ is a value and $[u/x]P$ is
a computation. Likewise, $[u/x]c$ is a command. As in $\mlm$, we use
derived syntactic classes of contexts as follows:
$$
\begin{array}{rcrclcrcrcl}
(\textrm{base contexts}) & \quad & L & ::= &
a\ehole\mid\letvc{\ehole}{v}{c}&\qquad&
(\textrm{cbn contexts}) & \quad & N & ::= & L\mid\subst{\ehole}{p}{c}.
\end{array}
$$
The result $N[P]$ of filling the hole of $N$ by a computation $P$ is
inherited from $\mlm$, and also the notion of structural
substitution $[N/a]t$, $[N/a]c$ and $[C/a]N'$ and, finally, the
definition of well-moded substitution $[u/x]N$ in cbn contexts that
yields cbn contexts.

\textbf{Reduction rules} of $\lvc$ are given in
Fig.~\ref{fig:CVCrules}, where co-variable $b$ is assumed to be
fresh in both $\beta$ rules. The first thing to check is that the
left-hand sides of the rules are well-formed expressions of $\lvc$
and that the respective right-hand sides belong to the same
syntactic categories. The second step consists in verifying subject
reduction: this is immediate for the rules other than $\beta$ since
they are just restrictions of reduction rules of $\mlm$, and it is
fairly easy to see that the right-hand sides of the $\beta$ rules
receive the same type as $P$.

\begin{figure}[tb]\caption{Reduction rules of $\lvc$}\label{fig:CVCrules}
$$
\begin{array}{rcrcl}
(\beta) &\quad& (\lambda v.P)V& \rightarrow& \mu b.\letvc{\ret
V}v{bP}\\&&(\lambda q.P)Q& \rightarrow& \mu b.\subst Qq{(bP)}\\
(\sigmabi)& \quad & \letvc{\ret V}vc&\rightarrow&[V/v]c\\&&
\subst Ppc&\rightarrow&[P/p]c\\
(\pi) & \quad & L[\mu a.c]&\rightarrow & [L/a]c\\
(\etam) & \quad & \mu a.aP &\rightarrow & P \quad(a\notin P)\\
(\eta_\letraw)&\quad&\letvc{P}{v}{a(\ret v)}&\rightarrow&aP
\end{array}
$$
\end{figure}

The rules $\beta$ and $\sigmabi$ are analogous to the respective
rules of $\lmmt$ (see further down in Section~\ref{subsec:lmmt}), where any execution of term substitution in the
reduction rules is delegated to an application of rule $\sigma$ and
where therefore the $\beta$-reduction rule of $\lmmt$ has a
right-hand side that is never a normal term.
They are ``lazy'' since they delay term
substitution, but, by putting together $\beta$, $\sigmabi$ and
$\etam$, we obtain the following derived \emph{eager} $\beta$ rules:
$$
(\betae)\qquad\qquad(\lambda v.P)V \rightarrow [V/v]P\qquad\qquad
(\lambda q.P)Q \rightarrow [Q/q]P.
$$
For the first rule, the derivation is
$$
(\lambda v.P)V\to_\beta \mu b.\letvc{\ret V}v{bP}\to_\sigmabi\mu
b.[V/v](bP) = \mu b.b([V/v]P)\to_\etam[V/v]P.
$$
For the second rule, it is analogous.
According to the form of $L$, the $\pi$-rule splits again into
$\pic$ and
$$
\begin{array}{rcrcl}
(\pil) & \quad & \letvc{\mu a.c}{v}{c'}&\rightarrow & [\letvc{\ehole}{v}{c'}/a]c.
\end{array}
$$

The calculus $\lvc$ is confluent. There are five critical pairs, each of
them corresponding to a critical pair of $\mlm$. A confluence proof can
be given using an abstract rewriting theorem
\cite{Dehornoy-vanOostrom2008} for $\beta$, $\sigma$ and $\pi$ and
strong commutation with the $\eta$-rules.

\subsection{Bind, let, and
substitution}\label{subseq:bind-let-subst}

In this section we
formally relate $\lvc$ with $\mlm$,
explaining the decompositions and refinements that the former brings
relatively to the latter.

As said, $\lvc$ is obtained from $\mlm$ by a three-fold restriction.
In the first step, function spaces are restricted to the
form $A\imp MA$. This already brings a novelty: the $\beta$ rule of
$\mlm$ can be decomposed into a new, finer-grained $\beta$-rule
\begin{equation}\label{eq:euler-beta}
(\lb x.t)u\to\mu a.\bind{\ret u}x{at}
\end{equation}
plus $\sigmabi$, $\etam$. Notice that in $\mlm$
(\ref{eq:euler-beta}) would break subject reduction, because $t$
would not be forced to have a monadic type.\footnote{One can marvel
how in (\ref{eq:euler-beta}) the constructors related to the type
$\imp$ in the l.h.s. of the rule are converted into an expression in
the r.h.s. using all of the constructors related to classical logic
and the monad.}

Let us call \emph{restricted} $\mlm$ this variant of $\mlm$, with
the restriction on function spaces and the variant
(\ref{eq:euler-beta}) of the $\beta$-rule. Then, $\lvc$ is clearly a
subsystem of restricted $\mlm$. Formally there is a forgetful map
$|\cdot|$ from the former to the latter that: (i) at the level of
types, forgets the distinction between value types and computation
types; (ii) at the level of expressions, forgets the distinction
between values and computations, and blurs the distinction between
$\letraw$ and substitution:
$$
|\letvc Pvc|=\bind {|P|}v{|c|}\qquad\qquad |\subst Ppc|=\bind{\ret
|P|}p{|c|}.
$$

What is the difference between $\letraw$ and substitution? This is
perhaps clearer in the intuitionistic subsystem of $\lvc$, which we
now spell out.

We follow the same steps as in our previous paper \cite{ourMSCS}, where the
intuitionistic subsystem of $\mlm$ (essentially Moggi's monadic
meta-language \cite{Moggi91}) was obtained. First we adopt a single
co-variable $*$ , say, which is never free in values or computations,
and which has a single free occurrence in commands. The
constructions $\mu *.c$ and $*P$ are like coercions between the
syntactic classes of computations and commands, coercions which in the
next step we decide not to write, causing the mutual inclusion of the
two classes, and the collapse of $\pic$ and of one of the cases of
$\pil$. The final step is to merge the two syntactic classes into a
single class of computations.

The resulting intuitionistic subsystem of $\lvc$ has the following
syntax:
$$
V,W \, ::= \, v\mid\lb x.P\qquad\qquad P,Q \, ::= \, p \mid \ret V
\mid Vu \mid \letvc PvQ \mid\subst PpQ.
$$
Again, $t,u\,::=\,V\mid P$ and $x,y\,::=\,v\mid p$. The reduction
rules are found in Fig.~\ref{fig:intuitionistic-red-rules}.

\begin{figure}[t]\caption{Reduction rules for the intuitionistic subsystem of
$\lvc$}\label{fig:intuitionistic-red-rules}
$$
\begin{array}{rcrcl}
(\beta) & \quad & (\lambda v.P)V& \rightarrow& \letvc{\ret V}v{P}\\
         & \quad & (\lambda q.P)Q& \rightarrow& \subst Qq{P}\\
(\sigmabi)& \quad & \letvc{\ret V}vQ&\rightarrow&[V/v]Q\\
          & \quad & \subst PqQ&\rightarrow&[P/q]Q\\
(\pil) & \quad & \letvc{\letvc PvQ}{w}{Q'}&\rightarrow & \letvc Pv{(Q;w.Q')}\\
       & \quad & \letvc{\subst PpQ}{w}{Q'}&\rightarrow & \subst Pp{(Q;w.Q')}\\
(\etal)&\quad&\letvc{P}{v}{\ret v}&\rightarrow&P
\end{array}
$$
where
$$
\begin{array}{rcl}
(\letvc PvQ);w.Q' &=& \letvc Pv{(Q;w.Q')}\\
(\subst PpQ);w.Q' &=& \subst Pp{(Q;w.Q')}\\
Q;w.Q' &=& \letvc Qw{Q'},\textrm { otherwise}
\end{array}
$$
\end{figure}

Back to the difference between $\letraw$ and substitution: the
$\sigma$ rule for $\letraw$ only substitutes values, while the
$\sigma$ rule for substitution substitutes any computation;
$\letraw$ enjoys $\pi$-rules, which are assoc-like rules for
sequencing the computation, and an $\eta$-rule, while substitution
does not. These distinct behaviors are amalgamated in the $\bindraw$
of $\mlm$.

\subsection{Continuations-monad instantiation}

The monad operation $\mon$ can be instantiated to be double negation
yielding the well-known continuations monad. We define an
instantiation that is capable of embedding $\lvc$ into
$\lb[\vbeta]$, the latter denoting simply-typed $\lambda$-calculus
with the only reduction rule $\vbeta$: $(\lambda x.t)V\to [V/x]t$
for values $V$, i.\,e., $V$ is a variable or
$\lambda$-abstraction.

For our purposes, the main role of the continuations-monad
instantiation is to provide a strict simulation, through which
strong normalization is inherited from $\lb[\vbeta]$. We also avoid
$\eta$-reduction in the target, so the instantiation makes use of
quite some \emph{$\eta$-expansions}
$$\up t:=\lambda x.tx,$$
with $x\notin t$. Clearly, this can only be done with terms that
will be typed by some implication later. $\lvc$ is
rather handy as source of such mapping, because of its distinction
between value variables that cannot be $\eta$-expanded and
computation variables that can.\footnote{One can define
continuations-monad instantiations on $\mlm$ that avoid
$\eta$-reduction on the target, but not uniformly on cbn and cbv
\cite{ourMSCS}. See Section \ref{sec:concl} for further
discussion.}.
The details are as follows.

We define a type $\cm A$ of
simply-typed $\lb$-calculus for every type $A$ of $\lvc$ ($\neg A$ is abbreviation for $A\imp\bot$ for some fixed type variable $\bot$ that will never be instantiated and hence qualifies as a type constant):\\[-4ex]
\begin{align*}
 \cm{X} & =X
 & \cm{(A\imp C)} & =\cm{A}\imp\cm{C}
 & \cm{(\mon B)} & =\neg\neg\cm{B}.\\[-4ex]
\end{align*}
Expressions $T$ of $\lvc$ are translated into terms $\cm T$ of
$\lb$-calculus, where an auxiliary definition of terms $\cmaux P$ for
computations $P$ of $\lvc$ is used. The idea is that $\cmaux P$ uses
$\eta$-expansions more sparingly than $\cm P$. The definition is in
Fig.~\ref{fig:refinst}, where we use an abbreviation $\Eta(t)=\lambda
k.kt$ with a fresh variable $k$ (the type of $\Eta(t)$ is the double negation of the type of its argument), and we assume that the co-variables $a$ of $\lvc$ are variables of the target $\lambda$-calculus (as we did in previous work on the continuations-monad instantiation of $\mlm$ \cite[Section 5.1]{ourMSCS}). Obviously, $\cm t$ and $\cmaux P$ are
always values of $\lb$-calculus, i.\,e., variables or
$\lb$-abstractions. $\cm P$ is even always a $\lb$-abstraction. (In the
whole development, we will never use that $\cmaux P$ is a value.)
\begin{figure}[tb]\caption{Continuations-monad instantiation}\label{fig:refinst}
\[
\begin{array}{r@{\,\,=\,\,}l@{\quad}r@{\,\,=\,\,}l}
 \cm{v}&v & \cm{(\ret V)}&\cmaux{(\ret V)}\\
 \cm{(\lambda x.P)}&\lambda x.\cm{P}& \mbox{$P\neq\ret V$: }\cm P&\lb k.\cmaux P(\up k)  \\[1ex]
 \cmaux{p}&p & \cm{(aP)}&\cmaux P(\up{a})\\
 \cmaux{(\ret V)}&\Eta(\cm V) & \cm{(\letvc Pvc)}&\cmaux P(\lambda v.\cm{c})\\
\cmaux{(\mu a.c)}&\lambda a.\cm{c}& \cm{(\subst Ppc)}&(\lb p.\cm c)\cm P\\[1ex]
\cmaux{(Vu)}&\lb k.\Eta(\cm u)\bigl(\lb w.\cm V w(\up k)\bigr)
\end{array}\]\vspace{-0.5cm}
\end{figure}
We define the type operation $\cmmin{(.)}$ by $\cmmin{(\mon B)}:=\neg\cm{B}$, which extends to co-contexts $\Delta$ by elementwise application.
We can easily check that the rules in Fig.~\ref{fig:refcmtypes} are
admissible. In general, if $P$ gets
type $A$, then $A$ is of the form $MB$, hence $\cm A=\neg\neg\cm B$.
If we then already know that $\cmaux P$ gets type $\cm A$, also $\cm
P$ gets that same type.
\begin{figure}[tb]\caption{Admissible typing rules for continuations-monad instantiation}\label{fig:refcmtypes}
$$
 \infer{\cm{\Gamma},\cmmin{\Delta}\vdash \cmaux{P}:\cm{A}}{
  \Gamma\vdash P:A\mid \Delta
 }
\qquad \infer{\cm{\Gamma},\cmmin{\Delta}\vdash \cm{t}:\cm{A}}{
  \Gamma\vdash t:A\mid \Delta
 }
\qquad
 \infer{\cm{\Gamma},\cmmin{\Delta}\vdash \cm{c}:\bot}{
  c:(\Gamma\vdash\Delta)
 }$$\vspace{-0.5cm}
\end{figure}

\begin{thm}[Strict simulation]\label{thm:simulation-for-refinedcm}
If $T\to T'$ in $\lvc$, then $\cm T\to^+_\vbeta \cm{T'}$ in
$\lb[\vbeta]$.
\end{thm}

\begin{cor}\label{cor:lvcsn} $\lvc$ is strongly normalizable and
confluent on typable expressions.
\end{cor}
\begin{proof}
Strong normalization is inherited from $\lambda[\vbeta]$ through strict
simulation. Confluence follows from strong normalizability and local
confluence.
\end{proof}


\section{Classical logic}\label{sec:classical-logic}

In this section we start by recalling the $\lmmt$-calculus, its main
critical pair, and its cbn and cbv fragments. Next we motivate and
develop a variant of $\lmmt$ with ``modes'', denoted $\lmmtvn$.
Finally, a monadic translation of $\lmmtvn$ into $\lvc$ gives to the
source system a semantics parameterized by a monad, and a proof of
confluence, through strong normalization, for the typed expressions.

\subsection{The $\lmmt$-calculus}\label{subsec:lmmt}

We recapitulate $\lmmt$. Expressions are defined by the following grammar.\\[-4ex]
\begin{align*}
 &(\textrm{values}) & V &::= x\mid\lambda x.t && (\textrm{co-values}) &
 E  & ::= a\mid u::e & & (\textrm{commands}) & c & ::= \com{t}{e}\\
 &(\textrm{terms}) & t,u & ::= V\mid \mu a.c && (\textrm{co-terms}) & e & ::= E\mid\mt x.c\\[-4ex]
\end{align*}
\noindent Expressions are ranged over by $T$, $T'$. Variables
(resp.~co-variables) are ranged over by $v$, $w$, $x$, $y$, $z$
(resp.~$a$, $b$). We assume a countably infinite supply of them and
denote any of them by using decorations of the base symbols.
Variable occurrences of $x$ in $\lambda x.t$ and $\mutil x.c$, and
$a$ in $\mu a.c$ are bound, and an expression is identified with
another one if the only difference between them is names of bound
variables.

Types are given by $A,B::=X\mid A\imp B$ with type variables $X$.
There is one kind of sequent per proper syntactic class
$\Gamma\vdash t:A\mid\Delta$ for terms, $\Gamma\mid e:A\vdash \Delta$ for
co-terms, and $c:(\Gamma\vdash\Delta)$ for commands, where $\Gamma$
ranges over consistent sets of variable declarations $x:A$ and
$\Delta$ ranges over consistent sets of co-variable declarations
$a:A$. Typing rules and reduction rules are
given in Fig.~\ref{fig:cbnlmmtnew},
\begin{figure}[tb]
\caption{Typing rules and reduction rules of
$\lmmt$}\label{fig:cbnlmmtnew}
$$
\infer{\Gamma,x:A\vdash x:A\mid\Delta}{}\qquad
\infer{\Gamma\vdash\lambda x.t:A\supset B\mid\Delta}{\Gamma,x:A\vdash
t:B\mid\Delta} \qquad \infer{\Gamma\vdash \mu
a.c:A\mid\Delta}{c:(\Gamma\vdash a:A,\Delta)}
$$
$$
\infer{\Gamma\mid a:A\vdash a:A,\Delta}{} \qquad
\infer{\Gamma\mid u::e:A\supset B\vdash\Delta}{\Gamma\vdash u:A\mid\Delta &
\Gamma\mid e:B\vdash \Delta} \qquad \infer{\Gamma\mid \mt
x.c:A\vdash\Delta}{c:(\Gamma,x:A\vdash\Delta)}
$$
$$
\infer{\com{t}{e}:(\Gamma\vdash\Delta)}{\Gamma\vdash t:A\mid \Delta &
\Gamma\mid e:A\vdash \Delta}
$$
$$
\begin{array}{rrcl@{\qquad\qquad}rrcl}
(\beta) & \com{\lambda x.t}{u::e}& \rightarrow & \com{u}{\mt
x.\com{t}{e}}
 & (\etamt) & \mt x.\com{x}{e} & \rightarrow & e,\textrm{ if $x\notin e$}\\
(\pi) & \com{\mu a.c}{e} & \rightarrow & [e/a]c
 & (\etam) & \mu a.\com{t}{a} & \rightarrow & t,\textrm{ if $a\notin t$}\\
(\sigma) & \com{t}{\mt x.c} & \rightarrow & [t/x]c\\
\end{array}
$$
\end{figure}
where we reuse the name $\beta$ of $\lb$-calculus (rule names are
considered relative to some term system), and the substitutions
$[e/a]$ and $[t/x]$ in expressions respecting the syntactic
categories are defined as usual. These are the reductions considered
by Polonovski \cite{PolonovskiFOSSACS2004};
however, the $\beta$-rule for
the subtraction connective is not included.

Following Curien and Herbelin \cite{CurienHerbelinICFP2000}, we
consider cbn and cbv fragments $\nlmm$ and $\vlmm$, respectively,
where the critical pair rooted in $\com{\mu a.c}{\mt x.c'}$ between
the rules $\sigma$ and $\pi$ is avoided. In $\nlmm$, we restrict the
$\pi$ rule to $\pin$, and dually in $\vlmm$, we restrict the
$\sigma$ rule to
$\sigmav$ as follows.
\[
\begin{array}{rrcl@{\qquad\qquad}rrcl}
(\pin) & \com{\mu a.c}{E} & \rightarrow & [E/a]c
& (\sigmav) & \com{V}{\mt x.c} & \rightarrow & [V/x]c\\[-2ex]
\end{array}
\]
In both fragments, the only critical pairs are trivial ones
involving $\etamt$ and $\etam$, hence $\nlmm$ and $\vlmm$ are
confluent since weakly orthogonal higher-order rewriting systems are
confluent as proved by van Oostrom and van Raamsdonk
\cite{vanOostrom-vanRaamsdonk94}.

\subsection{Towards a variant of $\lmmt$ with ``modes''}
Suppose we single out in $\lmmt$ a class of variables as \emph{value
variables}, ranged over by $v$. Let $n$ (resp.~$x$) range over the
non-value variables (resp.~both kinds of variables). Variables $n$
are called \emph{computation variables}, but the terminology
value/computation, like many decisions we will make, will get a full
justification only through the monadic semantics into $\lvc$ given
below. We call the distinction value/computation a \emph{mode}
distinction.

What syntactic consequences come from introducing a mode distinction
in variables? Quite some.

Since the bound variable in a $\lb$-abstraction is like a mode
annotation, also $u::e$ should come in two annotated versions, one
for each mode. The same is true of the type $A\imp B$. Since the
variable $n$ is not a value variable, it should not count as a
value. On the other hand, both versions of $u::e$ are co-values, but
what about $\mutil x.c$? In a kind of dual movement, since $n$ left
the class of values, $\mutil v.c$ enters the class of co-values (so
the only co-term that is not a co-value is $\mutil n.c$).

Revisiting the critical pair $\com{\mu a.c}{\mutil x.c'}$, it is
quite natural that the mode of $x$ resolves
the dilemma! In particular, the case $x=v$ gives a $\pi$-redex,
and it follows that we only need $\pi$-redexes where the right component of
the command is a co-value. On the other hand, the case $x=n$ gives a
$\sigma$-redex. In fact all commands of the form $\com t{\mutil
n.c'}$ are $\sigma$-redexes, but they do not cover yet another form
of $\sigma$-redex: $\com V{\mutil v.c'}$.
Do not forget the latter
does not cover $\com n{\mutil v.c'}$---this command is not a redex.

\subsection{$\lmmt$ with modes}

We now give the formal development of $\lmmt$ with modes, denoted
$\lmmtvn$.

\textbf{The expressions} of $\lmmtvn$ are inductively defined in
Fig.~\ref{fig:lmmtvn}. The names for value variables, computation
variables and both kinds of variables are those of $\lvc$.
\begin{figure}[tb]\caption{Expressions of $\lmmtvn$}\label{fig:lmmtvn}
\[
\begin{array}{rlcrll}
(\textrm{value vars}) & v,w,f&\quad&(\textrm{comp. vars}) & n,p,q\\
(\textrm{variables})& x,y,z::= v\mid p&&(\textrm{commands}) & c ::= \com te\\
(\textrm{values}) & V,W ::=  v\mid\lambda x.t&&(\textrm{co-values}) & E ::=  a\mid \mutil v.c\mid u::_xe\\
(\textrm{terms}) & t,u ::=  V\mid n\mid \mu a.c&&(\textrm{co-terms}) & e ::=  E\mid\mutil n.c\\
\end{array}
\]
\end{figure}
Variable $x$ gets a second role for denoting \emph{modes}: $x\in
\{v,n\}$. This allows to write $u::_xe$ and use variable $x$ in rules
governing $u::_v e$ and $u::_n e$ uniformly. Note that this is rather
a presentational device: there are only two modes, and they go by the
names $v$ and $n$. Then, $x$ in its second role is used to denote any of
these two modes.
In its first role, $x$ stands for one of the
countably many value variables, typically denoted by $v$, or one of the
countably many computation variables, typically denoted by $n$. In
using the name $x$ for both a variable and a mode, rules can be
written more succinctly because rule schemes comprising two rules get
the appearance of one single rule.

The separation between value and computation variables allows a mode
distinction in the proof expressions of $\lmmtvn$: values have value
mode, terms have computation mode. This will be fully justified by
the monadic semantics into $\lvc$ below, as values (resp.~terms)
will be mapped to values (resp.~computations) of the latter
calculus. Beware that neither the mode annotation in $u::_xe$ nor
the mode of the bound variable in a $\lb$-abstraction determines the
mode of the expression. In particular, contrary to the case of
$\lvc$, there is no need for a well-modedness constraint in the
definition of substitution. For instance, there is nothing wrong
with the operation $[\lb n.t/v]T$. A $\lb$-abstraction has value
mode, independently of the mode of the bound variable.

\textbf{Types} are formed from type variables $X$ by two
implications: $A\imp_vB$ and $A\imp_nB$. Generically, we may write
both implications as $A\imp_xB$.

Although implications carry a mode annotation, we refrain from
classifying them (let alone atomic types) with a mode.
As it will become clear from the monadic semantics to be introduced below,
we cannot determine from a type of $\lmmtvn$ alone whether its
semantics is a value or computation type; in fact, every type $A$ of
$\lmmtvn$ will determine a value type $\mdtr A$ and a computation
type $\mtr A$. In particular, $A\imp_xB$ determines both a value
type and a computation type, for both mode annotations $x$---even
though, of course, the annotation $x$ guides what those types are.
Hence, contrary to what happens in $\lvc$, we cannot expect in
$\lmmtvn$ that a syntactic category is attached to a particular type
mode, because in $\lmmtvn$ there is no such thing as type modes.
This is why the sequents in $\lmmtvn$ have the same forms as in
$\lmmt$ and carry \emph{no} well-modedness constraint. So,
declarations like $v:A\imp_nB$ or $n:A\imp_vB$ are perfectly normal.

The only typing rules of $\lmmtvn$ that differ from $\lmmt$ are
given in Fig.~\ref{fig:lmmtvn-typing}. Each of the two rules in that
figure stands for two rules that are uniformly written with
$x\in\{v,n\}$.
\begin{figure}[tb]
\caption{Typing rules of $\lmmtvn$ for the implications}\label{fig:lmmtvn-typing}
\[
\begin{array}{ccc}
\infer[R\mbox{-}\imp_x]{\Gamma\vdash\lambda x.t:A\imp_x
B\mid \Delta}{\Gamma,x:A\vdash t:B\mid \Delta}&
\quad&\infer[L\mbox{-}\imp_x]{\Gamma\mid u::_xe:A\imp_x B\vdash
\Delta}{\Gamma\vdash u:A\mid \Delta & \Gamma\mid e:B\vdash \Delta}
\end{array}
\]
\end{figure}

\textbf{The reduction rules} of $\lmmtvn$ given in
Fig.~\ref{fig:lmmtvn-rules} are copies of those of $\lmmt$, with a
\emph{moding constraint} in the $\beta$ rule and provisos in the
rules $\sigma_v$ and $\eta_{\mutil,v}$. The rule $\pi$ is restricted
to co-values in the spirit of rule $\pin$ of $\nlmm$. Note that
$\sigma_n$ reduction $\com{\mu a.c'}{\mt n.c} \rightarrow [\mu
a.c'/n]c$, and $\pi$ reduction $\com{\mu a.c}{\mt v.c'} \rightarrow
[\mt v.c'/a]c$ are both allowed in $\lmmtvn$.
In the rule $\eta_{\mt,x}$ with $x=v$, the co-term $e$ is restricted to
a co-value.  If we drop the condition, the co-value $\mt v.\com{v}{\mt
n.c}$ is reduced to $\mt n.c$ which is not a co-value.

\begin{figure}[tb]\caption{Reduction rules of $\lmmtvn$}\label{fig:lmmtvn-rules}
\[
\begin{array}{rrcl}
(\beta_x) & \com{\lambda x.t}{u::_xe}& \rightarrow & \com{u}{\mt
x.\com{t}{e}}\\
(\sigma_x) & \com{t}{\mt x.c} & \rightarrow & [t/x]c\quad\mbox{with: if $x=v$ then $t=V$}\\
(\pi) & \com{\mu a.c}{E} & \rightarrow & [E/a]c\\
(\etamtx) & \mt x.\com{x}{e} & \rightarrow & e,\textrm{ if $x\notin
e$ and: if $x=v$ then $e=E$}\\
(\etam) & \mu
a.\com{t}{a} & \rightarrow & t,\textrm{ if $a\notin t$}
\end{array}
\]
\end{figure}
The non-confluent critical pair of $\lmmt$ is avoided here for both
modes $x$.
\begin{center}
\begin{tabular}{c}
\xymatrix{[\mt x.c'/a]c&&\com{\mu a.c}{\mt
x.c'}\ar[rr]^{\,\,\,\,\,\sigma_x\,\,\mbox{\scriptsize only for
$x=n$}}\ar[ll]_{\pi\,\,\mbox{\scriptsize only for
$x=v$}\,\,\,}&&[\mu a.c/x]c'}
\end{tabular}
\end{center}
Thus, the reduction rules are weak enough to avoid the ``dilemma''
of $\lmmt$. On the other hand, the reduction rules may seem too weak
since command $\com n{\mutil v.c}$ is not a redex and not excluded by
typing.

There is a forgetful map $|\cdot|:\lmmtvn\to\lmmt$. It forgets the
distinctions between: value variables and computation variables; the
reduction rules $\beta_v$ and $\beta_n$, and similarly for the
reduction rules $\sigma$ and $\etamt$; the type constructors
$\imp_v$ and $\imp_n$; the typing rules $R$-$\imp_v$ and $R$-$\imp_n$,
and $L$-$\imp_v$ and $L$-$\imp_n$.

\subsection{Call-by-name and call-by-value}

Both the cbn and the cbv fragments $\nlmm$ and $\vlmm$ of $\lmmt$
can be embedded into $\lmmtvn$: variables are mapped into
computation variables and value variables, and $A\imp B$ is mapped
to $A\imp_nB$, and to $A\imp_vB$, respectively. Likewise, $u::e$ is
mapped to $u::_ne$ and $u::_ve$, respectively. Through these
embeddings $\lmmtvn$ becomes a conservative extension: on the images
of the translation, no new reductions arise w.\,r.\,t.~the source
calculi. Besides the two fragments, $\lmmtvn$ allows additionally
interaction between the cbv and cbn evaluation orders, without
losing (as we will see) the confluence property enjoyed by $\vlmm$
and $\nlmm$, but not by full $\lmmt$.

We have seen that $\sigmav$ and $\pin$ are adopted as two possible
solutions to the critical pair $\com{\mu a.c}{\mutil x.c}$. But now
we can see how drastic these solutions are.
We see that, in $\lmmtvn$, the command $\com t{\mutil n.c}$ is
always okay as a $\sigma$-redex (it never overlaps $\pi$),
 but, in $\lmmt$, that command is not considered as a $\sigmav$-redex when
$t$ is not a value. Likewise, in $\lmmtvn$, the command $\com{\mu
a.c}{\mutil v.c'}$ is okay as a $\pi$-redex (it never overlaps
$\sigma$), but, in $\lmmt$, that command does not count as a
$\pin$-redex.

\subsection{Monadic translation}\label{sec:montrans}

We now introduce a monadic translation of the system with modes into
$\lvc$. Since the system with modes embeds both $\vlmm$ and $\nlmm$,
the translation is uniform for cbn and cbv.

Using the abbreviation $\mtr A=\mon \mdtr A$, we recursively define
the value type $\mdtr A$ of $\lvc$ for each type $A$ of $\lmmtvn$
(and simultaneously obtain that $\mtr A$ is a computation type):
$$\mdtr X=X\qquad\mdtr{(A\imp_v B)}=\mdtr{A}\imp\mtr{B}\qquad\mdtr{(A\imp_n B)}=\mtr{A}\imp\mtr{B}.$$
For one binding $x:A$ in a term context $\Gamma$ of $\lmmtvn$, we define one
binding $\mmdtr{(x:A)}$ in a term context  of $\lvc$ as follows (one of
the two type operators $\mdtr{(.)}$ or $\mtr{(.)}$ is chosen, this is not a composition of operations):
$\mmdtr{(v:A)}:=v:\mdtr A$ and $\mmdtr{(n:A)}:=n:\mtr A$.
For an entire term context $\Gamma$, the operation is then done elementwise.  $\mtr{\Delta}$ is naturally defined by replacing every type $A$
in $\Delta$ by $\mtr A$.

The monadic translation of $\lmmtvn$ associates computations $\mtr
t$ with terms $t$, values $\mdtr V$ with values $V$, cbn contexts
$\mtr e$ with co-terms $e$ (which are even base contexts for
co-values $e$) and commands $\mtr c$ with commands $c$, and is given in
Fig.~\ref{fig:montrans}. Its crucial admissible typing rules are
found in Fig.~\ref{fig:montranstypes}.

Observe how, through the monadic translation, the differences
between $\mt v.c$ and $\mt n.c$, and between $u::_ve$ and $u::_ne$,
boil down to the difference between $\letraw$ and substitution in
$\lvc$.
\begin{figure}[tb]\caption{Monadic translation of $\lmmtvn$}\label{fig:montrans}
\[
 \begin{array}{rcl@{\qquad}rcl}
 \mtr V&=&\ret \mdtr{V}&\mdtr{v}&=&v\\
 \mtr n&=&n&\mdtr{(\lambda x.t)}&=&\lambda x.\mtr t\\
 \mtr{\mu a.c}&=&\mu a.\mtr c\\
 \mtr{a}&=&a\ehole&\mtr{\com te}&=&\mtr{e}[\mtr t]\\
 \mtr{\mt v.c}&=&\letvc{\ehole}v{\mtr{c}} &\mtr{\mt n.c}&=&\subst{\ehole}n{\mtr{c}}\\

 \mtr{u::_ve}&=&\letvc{\ehole}f{\letvc{\mtr{u}}w{\mtr{e}[fw]}}& \mtr{u::_ne}&=&\letvc{\ehole}f{\subst{\mtr{u}}q{\mtr{e}[fq]}}\\

 \end{array}
\]
\end{figure}

\begin{figure}[tb]\caption{Admissible typing rules for monadic translation of $\lmmtvn$}\label{fig:montranstypes}
$$
\infer{\mmdtr{\Gamma}\vdash \mtr t:\mtr A\mid \mtr{\Delta}}{\Gamma\vdash
t:A\mid \Delta} \qquad \infer{\mmdtr\Gamma\vdash\mdtr{V}:\mdtr
A\mid \mtr{\Delta}}{\Gamma\vdash V:A\mid \Delta} \qquad \infer{ \mtr
c:(\mmdtr\Gamma\vdash\mtr \Delta)}{c:(\Gamma\vdash\Delta)} \qquad
\infer{\mtr{e}[p]:(\mmdtr\Gamma,p:\mtr A\vdash
\mtr\Delta)}{\seql{\Gamma}{e}{A}{\Delta}}
$$
\end{figure}

\begin{thm}[Strict simulation]\label{thm:mtrgood}
  1. If $T\red T'$ in $\lmmtvn$, then $\mtr{T}\red^+\mtr{T'}$ in
  $\lvc$, where $T$, $T'$ are either two terms or two commands.

  2. If $e\red e'$ in $\lmmtvn$, then $\mtr{e}[P]\red^+\mtr{e'}[P]$ in
  $\lvc$ for any computation $P$ in $\lvc$.
\end{thm}

As a consequence, $\lmmtvn$ is the promised confluent calculus of
cut-elimination.

\begin{cor}$\lmmtvn$ is strongly normalizable and
confluent on typable expressions.
\end{cor}
\begin{proof}
Strong normalization is inherited from $\lvc$ (Corollary
 \ref{cor:lvcsn}) through strict simulation.
 Confluence follows from strong normalizability and local
confluence. The cornerstone of local confluence in $\lmmtvn$ is the
absence of overlap between $\sigma$ and $\pi$, as explained before.
\end{proof}


Through composition with the continuations-monad instantiation
$\cm{(\cdot)}:\lvc\to\lb[\vbeta]$, the monadic semantics
$\mtr{(\cdot)}:\lmmtvn\to\lvc$ is instantiated to a CPS semantics
$\cps{(\cdot)}:\lmmtvn\to\lb[\vbeta]$.

\begin{thm}[CPS translation]\label{thm:cpsgood}
  If $T\red T'$ in $\lmmtvn$, then $\cps{T}\red_{\vbeta}^+\cps{T'}$ in
  simply-typed $\lb$-calculus, where $T$, $T'$ are either two terms or two commands.
\end{thm}

\begin{proof} By putting together Thm.~\ref{thm:simulation-for-refinedcm} and
Thm.~\ref{thm:mtrgood}.
\end{proof}

So we are using the methodology of Hatcliff and Danvy
\cite{HatcliffDanvy94} to synthesize a new CPS translation, as done
in \cite{ourMSCS}.
The obtained CPS translation is easily produced,
but its explicit typing behaviour and recursive structure is rather
complex (no space for details).
Given that the monadic translation is uniform for cbn and cbv, so is
the CPS translation. Given that the monadic translation and the
continuations-monad instantiation produce strict simulations, the
CPS translation embeds $\lmmtvn$ into the simply-typed
$\lb$-calculus.


\section{Final remarks}\label{sec:concl}

\textbf{Recovering confluence in classical logic.} Let us return to
Fig.~\ref{fig:new-picture} and ignore the monad instantiation. In
this paper two \emph{confluent} systems are proposed where the cbn
and cbv fragments of $\lmmt$ embed: $\lmmt$ with modes and $\lvc$.
As usual (recall linear and polarized logics
\cite{DanosJoinetSchellinxJSL97,ZeilbergerAPAL08}), confluence is
regained through refinement/decoration of the logical connectives of
classical logic. In the case of $\lmmt$ with modes, the
``amalgamation'' of cbn and cbv is obtained through mode distinctions and
annotations; in the case of $\lvc$, the distinction between value
and computation expressions and types is done on top of an already
refined system ($\mlm$), where classical logic is enriched with a
monad. The forgetful map from $\lmmt$ with modes to full $\lmmt$
forgets about modes with loss of confluence
\cite{CurienHerbelinICFP2000}, whereas the forgetful map from $\lvc$
to $\mlm$ blurs the distinction value/computation without loss of
confluence \cite{ourMSCS}.

The many ways out of the $\sigma/\pi$-dilemma illustrate the general
theme of the missing information in classical cut-elimination.
In the system $LK^{tq}$ of \cite{DanosJoinetSchellinxJSL97}, the
extra information that drives the cut-elimination procedure is the
``color'' of the cut formula. In $\lmmt$ the syntax of formulas is
not enriched, so the two ways out of the dilemma make use of other
means of expression (whether the term $t$ in $\com{t}{\mutil x.c}$
(resp. co-term $e$ in $\com{\mu a.c}{e}$) is a value (resp. a
co-value)). In $\lmmtvn$ the extra information is simply provided by
the mode of a variable.

Although we tried to give a self-contained presentation of
$\lmmtvn$, with a later justification through the monadic semantics,
the semantics appeared before the syntax:
first we designed $\lvc$ and then we ``pulled back''
to the syntax of $\lmmt$ an abstraction of the design of $\lvc$. The
resulting system is striking for many reasons. By amazingly simple
means, $\lmmtvn$ resolves the cbn/cbv dilemma while still
comprehending the cbn and cbv fragments.
Polarized systems achieve
the same kind of goals, but by rather more elaborate means:
co-existence of positive and negative fragments, mediated by
``shift'' operations (see \cite{ZeilbergerAPAL08}, or the long
version of \cite{Munch-MaccagnoniCSL09}, available from the author's
web page). On the other hand, the absence of type modes makes $\lmmtvn$ rather less structured than colored or polarized systems. As
a conclusion, $\lmmtvn$ proves that one does not need a very
elaborate proof-theoretical analysis in order to ``fix'' classical
logic.

We now comment on two subjects to which we made lateral
contributions.

\textbf{Calculi of values and computations.} In the literature there
are other \emph{intuitionistic} calculi of values and computations,
for instance Filinski's multi-monadic meta-language (M3L)
\cite{FilinskiTCS2007} and Levy's call-by-push-value (CBPV)
\cite{LevyHOSC2006}. These are very rich languages, whose typing
systems include products and sums; in addition, M3L lets monads be
indexed by different ``effects'' and allows ``sub-effecting'',
whereas CBPV decomposes the monad into two type operations $U$ and
$F$. Notwithstanding this, the main difference of these languages
from the intuitionistic $\lvc$ is that function spaces are
computation types, and therefore $\lambda$-abstractions are
computations. This classification has a denotational justification:
in M3L and CBPV a computation type is a type that denotes a
$\mathcal T$-algebra (for $\mathcal T$ some semantic monad); it
follows that $A\imp C$ is a computation type, for if $C$ denotes one
such algebra, so does $A\imp C$. On the other hand, our
classification of types into value types and computation types
follows the suggestion by Hatcliff and Danvy \cite{HatcliffDanvy94},
and results in a system where ``values'' (=terms that receive a
value type) are simultaneously values (=fully evaluated expressions)
in the traditional sense of operational semantics.

\textbf{Generic account of CPS translations.} The idea of factoring
CPS translations into a monadic translation and a ``generic''
instantiation to the continuations monad is due to Hatcliff and
Danvy \cite{HatcliffDanvy94}. The extension of this idea to CPS
translation of classical source systems is found in the authors'
previous work \cite{ourMSCS}, where the system $\mlm$ was
introduced. The results from \emph{op.\,cit.}~are
illustrated in Fig.~\ref{fig:MSCS}, which in fact contains two
pictures, one for each of the cbn and cbv fragments of $\lmmt$. Each
fragment required its own monadic translation and optimized
instantiation in order to achieve strict simulation by $\vbeta$ in the
$\lb$-calculus.

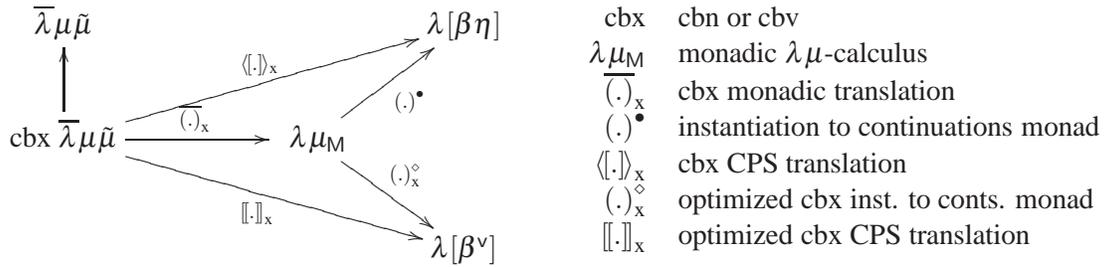
\begin{figure*}[t]\caption{Cbx pictures from \cite{ourMSCS}}\label{fig:MSCS}
\begin{center}
\begin{tabular}{ccc}
\xymatrix{%
\lmmt&&&{\lb[\beta\eta]}\\
\textrm{cbx }\lmmt\ar[u]\ar[rr]^{\;\;\;\vmo{(.)}_{\textrm x}}\ar[drrr]_{{[\![.]\!]}_{\textrm x}}\ar[urrr]^{\;\;\langle\![.]\!\rangle_{\textrm x}}
&&\;\;\mlm\ar[dr]^{{(.)}^{\diamond}_{\textrm x}}\ar[ur]_{\;\cm{(.)}}&\\
&&&{\lb[\vbeta]}\\
}%
&\qquad&
\begin{tabular}[t]{rl}
cbx&cbn or cbv\\
$\mlm$&monadic $\lm$-calculus\\
$\vmo{(.)}_{\textrm x}$&cbx monadic translation\\
$\cm{(.)}$&instantiation to continuations monad\\
${\langle\![.]\!\rangle}_{\textrm x}$&cbx CPS translation\\
${(.)}^{\diamond}_{\textrm x}$&optimized cbx inst. to conts. monad\\
${[\![.]\!]}_{\textrm x}$&optimized cbx CPS
translation
\end{tabular}
\end{tabular}
\end{center}\vspace{-0.5cm}
\end{figure*}

As a by-product of the present paper, we obtain an improvement in
the generic account of CPS translations from classical source
systems. Relatively to the results in our paper \cite{ourMSCS} one
sees the following improvements: (i) A single monadic translation
treats uniformly cbn and cbv; its source grew and its target shrunk,
relatively to the monadic translations for the cbn and cbv fragments
\cite{ourMSCS}. (ii) The instantiation with continuations monad
works for both cbn and cbv without dedicated optimizations.


\paragraph{Acknowledgments.}
We thank our anonymous referees for their helpful comments. Jos\'{e}
Esp\'{\i}rito Santo and Lu\'{\i}s Pinto have been financed by the Research
Centre of Mathematics of the University of Minho with the Portuguese
Funds from the "Funda\c{c}\~{a}o para a Ci\^{e}ncia e a Tecnologia", through the
Project PEstOE/MAT/UI0013/2014.



\bibliographystyle{eptcs}
\bibliography{vc}

\end{document}